\documentclass[11pt,english,numbers=endperiod]{article}

\usepackage{amsmath,amsfonts,amssymb,amsthm}
\usepackage{graphicx}
\usepackage{enumerate}
\usepackage{subfigure}

\usepackage[hyperindex,breaklinks]{hyperref}
\hypersetup{
  colorlinks   = true, 
  urlcolor     = blue, 
  linkcolor    = blue, 
  citecolor    = black 
}

\usepackage[a4paper]{geometry}
\newtheorem{theorem}{Theorem}

\newtheorem{proposition}{Proposition}

\newtheorem{assumption}{Assumption}

\newtheorem{example}{Example}
\newtheorem*{theorem*}{Theorem}
\theoremstyle{remark}
\newtheorem{remark}{Remark}

\makeatletter
\renewcommand*{\@seccntformat}[1]{%
  \csname the#1\endcsname.
}
\renewcommand\section{\@startsection {section}{1}{\z@}
                                   {-3.5ex \@plus -1ex \@minus -.2ex}%
                                   {2.3ex \@plus.2ex}%
                                   {\centering\normalsize\scshape}}
\renewcommand{\subsection}[1]{%
  \pagebreak[2]
  \refstepcounter{subsection} \addcontentsline{toc}{subsection}{
    {\protect\makebox[0.3in][r]{\thesubsection.} #1}}
  \noindent
  \textbf{\thesubsection.\ {#1}.}
  \everypar={}%
}
\renewcommand{\subsubsection}[1]{%
  \pagebreak[2]
  \refstepcounter{subsubsection} \addcontentsline{toc}{subsubsection}{
    {\protect\makebox[0.3in][r]{\thesubsubsection.} #1}}
  \noindent
  \textbf{\thesubsubsection.\ {#1}.}
  \everypar={}%
}

\makeatother
\usepackage{xspace}
\newcommand{\xmath}[1]{\ensuremath{#1}\xspace}
\newcommand{\E}{\xmath{\mathbb{E}}}
\renewcommand{\Pr}{\xmath{\mathbb{P}}}
\newcommand{\diff}{{\mathnormal{d}}}
\newcommand{\bmath}[1]{\mbox{\boldmath{$#1$}}}

\begin{document}
\title{\textbf{Incorporating Views on Marginal Distributions in the
    Calibration of Risk Models}} 

\author{Santanu Dey \and Sandeep Juneja\thanks{Corresponding author.
    Email addresses: $\{$dsantanu, juneja, kamurthy$\}$@tifr.res.in}
  \and Karthyek R. A. Murthy}

\date{\textit{Tata Institute of Fundamental Research,
    Mumbai}\\[\baselineskip]} 

\maketitle

\begin{abstract}
  Entropy based ideas find wide-ranging applications in finance for
  calibrating models of portfolio risk as well as options pricing.
  The abstracted problem, extensively studied in the literature,
  corresponds to finding a probability measure that minimizes relative
  entropy with respect to a specified measure while satisfying
  constraints on moments of associated random variables.  These
  moments may correspond to views held by experts in the portfolio
  risk setting and to market prices of liquid options for options
  pricing models. However, it is reasonable that in the former
  settings, the experts may have views on tails of risks of some
  securities. Similarly, in options pricing, significant literature
  focuses on arriving at the implied risk neutral density of benchmark
  instruments through observed market prices. With the intent of
  calibrating models to these more general stipulations, we develop a
  unified entropy based methodology to allow constraints on both
  moments as well as marginal distributions of functions of underlying
  securities.  This is applied to Markowitz portfolio framework, where
  a view that a particular portfolio incurs heavy tailed losses is
  shown to lead to fatter and more reasonable tails for losses of
  component securities. We also use this methodology to price
  non-traded options using market information such as observed option
  prices and implied risk neutral densities of benchmark instruments.
\end{abstract}

\section{Introduction}
Entropy based ideas have found a number of popular applications in
finance over the last two decades. A key application involves
portfolio optimization where we often have a prior probability model
and some independent expert views on the assets involved. If such
views are of the form of constraints on moments, entropy based methods
are used (see, e.g., Meucci \cite{fully:flexible}) to arrive at a
`posterior' probability measure that is closest in the sense of
minimizing relative entropy or $I$-divergence to the prior probability
model while satisfying those moment constraints. Another important
application involves calibrating the risk neutral probability measure
used for pricing options (see, e.g., Buchen and Kelly
\cite{buch:kelly}, Stutzer \cite{stutzer}, Avellaneda et
al. \cite{avel3}).  Here, entropy based ideas are used to arrive at a
probability measure that correctly prices given liquid options (which
are expectations of option payoffs) while again being closest to a
specified prior probability measure.

As indicated, in the existing literature the conditions imposed on the
posterior measure correspond to constraints on the moments of the
underlying random variables. However, the constraints that arise in
practice may be more general.  For instance, in portfolio optimization
settings, an expert may have a view that a certain index of stocks has
a fat-tailed t-distribution, and is looking for a posterior joint
distribution as a model of stock returns that satisfies this
requirement while being closest to a prior model, that may, for
instance, be based on historical data. 

Similarly, a view on the risk neutral density of a certain financial
instrument would also be reasonable if it is heavily traded, e.g.,
futures contract on a market index, and such views on marginal
densities can be used to better price less liquid instruments that are
correlated with the heavily traded instrument. There is now a sizable
literature that focuses on estimating the implied risk neutral density
from the observed option prices of an asset that has a highly liquid
options market (see \cite{jackwerth2004option} for a comprehensive
review).  In \cite{figlewski2009estimating}, Figlewski notes that the
implied risk neutral density of the US market portfolio, as a whole
entity, implicitly captures market's expectations, investors' risk
preferences and sensitivity to information releases and events. This
is usually not possible with just a finite number of constraints on
expected values of payoffs from options.  So in the options pricing
scenario, views on the posterior measure could include, for example,
those on the the implied risk neutral density of a security price
estimated from certain heavily traded options written on that
security. See, for example, Avellaneda \cite{avel1} for a discussion
on the need to use all the available econometric information and
stylized market facts to accurately calibrate mathematical models.

Motivated by these considerations, in this paper we devise a
methodology to arrive at a posterior probability measure when the
constraints on this measure are of a general nature that, apart from
moment constraints, include specifications of marginal distributions
of functions of underlying random variables as well.

\vspace{0.1in}

\noindent {\em Related literature:} 
The evolving literature on updating models for portfolio optimization
to include specified views builds upon the pioneering work of Black
and Litterman \cite{black:litterman}. They consider variants of
Markowitz's model where the subjective views of portfolio managers are
used as constraints to update models of the market using ideas from
Bayesian analysis. Their work focuses on Gaussian framework with views
restricted to linear combinations of expectations of returns from
different securities.  Since then a number of variations and
improvements have been suggested (see, e.g., \cite{mina:xiao},
\cite{least:disc} and \cite{qian:gorman}).  
Earlier, Avellaneda et al. \cite{avel3} used weighted Monte Carlo
methodology to calibrate asset pricing models to market data (also see
Glasserman and Yu \cite{glasserman:yu}).  Buchen and Kelly in
\cite{buch:kelly} and Stutzer in \cite{stutzer} use the entropy
approach to calibrate one-period asset pricing models by selecting a
pricing measure that correctly prices a set of benchmark instruments
while minimizing $I$-divergence from a prior specified model, that
may, for instance be estimated from historical data (see the recent
survey article \cite{kitamura:stutzer}).

\vspace{0.1in}

\noindent {\em Our contributions:} 
As mentioned earlier, we focus on examples related to portfolio
optimization and options pricing. It is well known that for views
expressed as a finite number of moment constraints, the optimal
solution to the $I$-divergence minimization can be characterized as a
probability measure obtained by suitably exponentially twisting the
original measure; this exponentially twisted measure is known in
literature as the Gibbs measure (see, for instance,
\cite{dembo:zeit}).  We generalize this to allow cases where the
expert views may specify marginal probability distribution of
functions of random variables involved. We show that such views, in
addition to views on moments of functions of underlying random
variables can be easily incorporated. In particular, under technical
conditions, we characterize the optimal solution with these general
constraints, when the objective is $I$-divergence and show the
uniqueness of the resulting optimal probability measure.

As an illustration, we apply our results to portfolio modeling in
Markowitz framework where the returns from a finite number of assets
have a multivariate Gaussian distribution and expert view is that a
certain portfolio of returns is fat-tailed. We show that in the
resulting probability measure, under mild conditions, all correlated
assets are similarly fat-tailed.  Hence, this becomes a reasonable way
to incorporate realistic tail behavior in a portfolio of
assets. Generally speaking, the proposed approach may be useful in
better risk management by building conservative tail views in
mathematical models. We also apply our results to price an option
which is less liquid and written on a security that is correlated with
another heavily traded asset whose risk neutral density is inferred
from the options market prices. We conduct numerical experiments on
practical examples that validate the proposed methodology.

\textit{Organization of the paper:} We formulate the model selection
problem as an optimization problem in Section \ref{SEC-MSP}, and
derive the posterior probability model as its solution in Section
\ref{SEC-OPT-SOl}. In Section \ref{SEC-PORTFOLIO-APP}, we apply our
results to the portfolio problem in the Markowitz framework and
develop explicit expressions for the posterior probability measure.
There we also show how a view that a portfolio of assets has a
`regularly varying' fat-tailed distribution renders a similar
fat-tailed marginal distribution to all assets correlated to this
portfolio. Further, we numerically test our proposed algorithms on
practical examples.  In Section \ref{SEC-OPT-APP}, we illustrate the
applicability of the proposed framework in options pricing
scenario. Finally, we end in Section \ref{SEC-CONC} with a brief
conclusion.  All but the simplest proofs are relegated to the
Appendix.

\section{The Model Selection Problem}
\label{SEC-MSP}
In this section, we briefly review the notion of relative entropy
between probability measures and use it to formally state our model
selection
problem.\\


\subsection{Relative Entropy and its Variational Representation}
Let $(\Omega,\mathcal{F})$ denote a measurable space and $\mathcal{P}$
denote the set of all probability measures on $(\Omega,\mathcal{F}).$
If a measure $\nu$ on $(\Omega,\mathcal{F})$ is absolutely continuous
with respect to $\mu,$ we denote this by $\nu \ll \mu.$ For any
$\nu,\mu \in\mathcal{P},$ the \textit{relative entropy} of $\nu$ with
respect to $\mu$ (also known as \textit{$I$-divergence} or
\textit{Kullback-Leibler divergence}) is defined as
\[D(\nu||\mu) :=
\begin{cases}
  \int \log\left(\frac{\diff \nu}{\diff \mu}\right)\diff \nu,
  &\text{ if } \nu \ll \mu, \\
  \infty, &\text{ otherwise.}
\end{cases}
\]
For any bounded measurable function $\psi$ mapping $\Omega$ into
$\mathbb{R},$ it is well known that,
\begin{equation}
  \log \int_\Omega e^\psi d \mu  =\sup_{\nu} \left\{ \int_\Omega
    \psi\,d\nu-D(\nu\mid\mid\mu) \right\}.  
  \label{VAR-REP-REL-ENT}
\end{equation}
Furthermore, this supremum is attained at $\nu^*$ given by:
\begin{equation} 
  \frac{d\nu^*}{d\mu}=\frac{e^\psi}{\int e^\psi\,d\mu}.
  \label{eqn:opt_measure}
\end{equation}
See, for instance, \cite{dup:eli} for a proof of this and for
other concepts related to relative entropy.\\

\subsection{Problem Formulation}
Let the random vectors ${\bf X} = (X_1,\ldots,X_m)$ and ${\bf Y} =
(Y_1,\ldots,Y_n)$ denote the risk factors associated with the prior
reference risk model, which is specified as a joint probability
density $f({\bf x,y})$ over ${\bf X}$ and ${\bf Y}.$ This model,
typically arrived using statistical analysis of historical data, is
used for risk analysis (such as calculating expected shortfall, VaR,
etc.), or for choosing optimal positions in portfolios. However, the
market presents itself with additional information, usually in the
form of `views' of experts (or) current market observables. These
views can be simple moment constraints as in,
\begin{equation*}
  \int_{\bf x,y} h_i({\bf x,y}) \Pr(\diff {\bf x},\diff {\bf y})  =
  c_i, i=1,\ldots,k, 
  \label{MOM-CONST}
\end{equation*}
(or) as detailed as constraints over marginal densities:
\begin{equation*}
  \int_{\bf y} \Pr(\diff {\bf x},\diff {\bf y}) = g({\bf x})  
  \text{ for all {\bf x}}, 
  \label{MARG-CONST}
\end{equation*}
where $c_i,i=1,\ldots,k$ are constants, $g(\cdot)$ is a given marginal
density of ${\bf X}$ and $\Pr(\cdot)$ is the unknown probability
measure governing the risk factors. Then our objective is to identify
a probability model that has minimum relative entropy with respect to
the prior model $f(\cdot,\cdot)$ while agreeing with the views on
moments of ${\bf Y}$ and marginal distribution of ${\bf X}.$ Though
relative entropy $D(\cdot||\cdot)$ is not a metric, it has been widely
used to discriminate between probability measures in the context of
model calibration (see \cite{thomas:cover}, \cite{buch:kelly},
\cite{stutzer},\cite{avel2}, \cite{avel1}, \cite{avel3},
\cite{fully:flexible} and \cite{kitamura:stutzer}). Let
$\mathcal{P}(f)$ denote the collection of probability density
functions which are absolutely continuous with respect to the density
$f(\cdot,\cdot)$ (a density $\tilde{f}(\cdot,\cdot)$ is said to be
absolutely continuous with respect to $f(\cdot, \cdot)$ if for almost
every $x$ and $y$ such that $f(x,y) = 0,$ $\tilde{f}(x,y)$ also equals
$0$). Formally, the resulting optimization problem ${\bf O_1}$ is:
\[\min_{\tilde{f} \in \mathcal{P}(f)} \int \log
\left(\frac{\tilde{f}(\bold x,\bold y)}{f(\bold x,\bold y)}\right)
\tilde{f}(\bold x,\bold y) d \bold x d \bold y,\] subject to: 
\begin{subequations}
\begin{align}
  \int_{\bold y} \tilde{f}(\bold x,\bold y) d \bold y = g(\bold
  x)\,\,\,\text{for all}\,\,\,\bold x, \text{ and
  } \label{eqn:constr_3}\\
  \int_{\bold x,\bold y} h_i(\bold x, \bold y) \tilde{f}(\bold x,\bold
  y) d \bold x d \bold y = c_i, i=1,2\ldots,k. \label{eqn:constr_4}
\end{align}
\end{subequations}

\section{Solution to the Optimization Problem ${\bf O_1}$} 
\label{SEC-OPT-SOl}
Some notation is needed to proceed further.  For any $\bmath{\lambda}=
(\lambda_1,\lambda_2,\ldots,\lambda_k) \in \mathbb{R}^k,$
let \[f_{\bmath{\lambda}}(\bold y|\bold x):= \frac{\exp
  (\sum_{i=1}^k\lambda_i h_i(\bold x,\bold y)){f}(\bold y|\bold
  x)}{\int_{\bold y} \exp (\sum_{i=1}^k\lambda_i h_i(\bold x,\bold y))
  {f}(\bold y|\bold x)\diff \bold y}=\frac{\exp (\sum_{i=1}^k\lambda_i
  h_i(\bold x,\bold y)){f}(\bold x,\bold y)}{\int_{\bold y} \exp
  (\sum_{i=1}^k\lambda_i h_i(\bold x,\bold y)) {f}(\bold x,\bold y)
  \diff \bold y}\] whenever the denominator exists. Further, let
$f_{\bmath{\lambda}}(\bold x,\bold y):=f_{\bmath{\lambda}}(\bold
y|\bold x)\times g(\bold x)$ denote a joint density function of
$(\bold X, \bold Y)$ and $\E_{\bmath{\lambda}}[\cdot]$ denote the
expectation under $f_{\bmath{\lambda}}(\cdot,\cdot)$. Let $m_g(\cdot)$
be the measure corresponding to the probability density $g(\cdot)$ on
$\mathbb{R}^m.$ For a mathematical claim that depends on ${\bf x} \in
\mathbb{R}^m,$ say $S({\bf x}),$ we write $S({\bf x})$ for almost all
${\bf x},$ with respect to $g({\bf x})\diff {\bf x}$ to mean that
$m_g(\{{\bf x}: S({\bf x}) \text{is false}\})=0.$
\begin{theorem}
  If there exists a $\bmath{\lambda} =
  (\lambda_1,\ldots,\lambda_k)\in\mathbb{R}^k$ such that
  \begin{enumerate}[(a)]
  \item $\int_{\bold y} \exp (\sum_i\lambda_i h_i({\bf x,y}))
    {f}(\bold x,\bold y)d \bold y<\infty$ for almost all $\bold x$
    with respect to $g(\bold x)d \bold x$, and
  \item $E_{\bmath{\lambda}}[h_i(\bold X,\bold Y)]=c_i,$ for
    $i=1,\ldots,k,$
  \end{enumerate}
  then $f_{\bmath{\lambda}}(\cdot)$ is an optimal solution to the
  optimization problem ${\bf O_1}$.
  \label{THM-OPT-SOL}
\end{theorem}
\noindent It is natural to ask for conditions under which such a
$\bmath{\lambda} = (\lambda_1,\ldots,\lambda_k)$ exists. Here, we
provide a simple condition in Remark \ref{EXISTENCE-MAIN} below. A
further set of elaborate conditions can be found in Theorem 3.1 of
\cite{csiszar:1975}.
\begin{remark}[On the existence of $\bmath{\lambda}$]
  \label{EXISTENCE-MAIN}
  For every $(c_1,\ldots,c_k)$ in the interior of convex hull of the
  support of the probability density function induced on
  $\mathbb{R}^k$ by the mapping $({\bf x}, {\bf y}) \longmapsto
  (h_1({\bf x},{\bf y}), \ldots, h_k({\bf x},{\bf y})),$ it follows
  from Theorem 3.1 of \cite{csiszar:1975} that there exists a
  $(\lambda_1,\ldots,\lambda_k)$ satisfying conditions of Theorem
  \ref{THM-OPT-SOL}.
\end{remark}

\begin{proof}[Proof of Theorem \ref{THM-OPT-SOL}]
  In view of (\ref{eqn:constr_3}), we may fix the marginal
  distribution of $\bold X$ to be $g(\bold x)$ and re-express the
  objective as
  \[\min_{\tilde{f}(\cdot |\bold x) \in {\cal P}({f}(\cdot |\bold x)),
    \forall \bold x} \int_{\bold x,\bold y}
  \log\left(\frac{\tilde{f}(\bold y|\bold x)}{f(\bold y|\bold
      x)}\right)\tilde{f}(\bold y|\bold x) g(\bold x) d \bold y d
  \bold x+\int_{\bold x} \log\left(\frac{g(\bold x)}{f(\bold
      x)}\right)g(\bold x)d \bold x\,.\] The second integral is a
  constant and can be dropped from the objective. The first integral
  can be expressed as
  \[\int_{\bold x} \min_{\tilde{f}(\cdot
    |\bold x) \in {\cal P}({f}(\cdot |\bold x))} \left (\int_{\bold y}
    \log \frac{\tilde{f}(\bold y|\bold x)}{f(\bold y|\bold x)}
    \tilde{f}(\bold y|\bold x) d \bold y \right )g(\bold x) d \bold
  x\,.\] Similarly the moment constraints can be re-expressed as
  \[\int_{\bold x,\bold y} h_i(\bold x,\bold y) \tilde{f}(\bold
  y|\bold x)g(\bold x) d \bold x d \bold y = c_i,\,\,\,i=1,2,...,k.\]
  This, in turn, is same as:
  \[\int_{\bold x}\left (\int_{\bold y} h_i(\bold x,\bold y)
    \tilde{f}(\bold y|\bold x)d \bold y \right )g(\bold x)d \bold x =
  c_i,\,\,\,i=1,2,...,k\,.\] Then, the Lagrangian for this $k$
  constraint problem is,
  \[\int_{\bold x} \left [ \min_{\tilde{f}(\cdot |\bold x) \in {\cal
        P}({f}(\cdot |\bold x))} \int_{\bold y} \left ( \log
      \frac{\tilde{f}(\bold y|\bold x)}{f(\bold y|\bold x)}
      \tilde{f}(\bold y|\bold x) -\sum_{i=1}^k\delta_i h_i(\bold
      x,\bold y)\tilde{f}(\bold y|\bold x) \right)d \bold y \right ]
  g(\bold x) d \bold x+\sum_{i=1}^k\delta_i c_i,\] for $\delta_i \in
  \mathbb{R}.$ Note that by (\ref{eqn:opt_measure}),
  \[\min_{\tilde{f}(\cdot |\bold x) \in {\cal P}({f}(\cdot |\bold x))}
  \int_{\bold y} \left ( \log \frac{\tilde{f}(\bold y|\bold
      x)}{f(\bold y|\bold x)} \tilde{f}(\bold y|\bold x)
    -\sum_i\delta_i h_i(\bold x,\bold y)\tilde{f}(\bold y|\bold x)
  \right )dy\] has the solution
  \[f_{\bmath{\delta}}(\bold y|\bold x)= \frac{\exp (\sum_i\delta_i
    h_i(\bold x,\bold y)){f}(\bold y|\bold x)}{\int_{\bold y} \exp
    (\sum_i\delta_i h_i(\bold x,\bold y)) {f}(\bold y|\bold x)d \bold
    y}=\frac{\exp (\sum_i\delta_i h_i(\bold x,\bold y)){f}(\bold
    x,\bold y)}{\int_{\bold y} \exp (\sum_i\delta_i h_i(\bold x,\bold
    y)) {f}(\bold x,\bold y)d \bold y},\] where we write
  $\bmath{\delta}$ for $(\delta_1,\delta_2,\ldots,\delta_k)$.  Now
  taking $\bmath{\delta}=\bmath{\lambda}$, it follows from the
  Assumptions (a) and (b) in the statement of Theorem
  \ref{THM-OPT-SOL} that $f_{\bmath{\lambda}}(\bold x,\bold
  y)=f_{\bmath{\lambda}}(\bold y|\bold x)g(\bold x)$ is a solution to
  the optimization problem ${\bf O_1}\,.$ \qedhere
\end{proof}
\noindent In Theorem \ref{thm:exist_unique}, we give conditions that
ensure uniqueness of a solution to the optimization problem ${\bf O_1}$
whenever it exists. The proof of Theorem \ref{thm:exist_unique} is
presented in the Appendix.
\begin{theorem}
  Suppose that for almost all ${\bold x}$ w.r.t. $g(\bold x)d \bold
  x$, conditional on $\bold X=\bold x$, the random variables
  $h_1(\bold x,\bold Y),h_2(\bold x,\bold Y),\ldots,h_k(\bold x,\bold
  Y)$ are linearly independent. Then, if a solution to the constraint
  equations
  \[E_{\bmath{\lambda}}[h_i(\bold X,\bold Y)]=c_i, i=1,\ldots,k\]
  exists, it is unique.
  \label{thm:exist_unique}
\end{theorem}

\begin{remark}
  \label{change:var} 
  Theorem \ref{THM-OPT-SOL}, as stated, is applicable when the updated
  marginal distribution of a sub-vector $\bold X$ of the given random
  vector $(\bold X,\bold Y)$ is specified.  More generally,
  constraints on marginal densities and moments of functions of the
  given random vector can also be incorporated by a routine change of
  variable technique. This is illustrated below:

  Let $\bold Z=(Z_1,Z_2,\ldots,Z_N)$ denote a random vector taking
  values in $S \subseteq \mathbb{R}^N$ and having a (prior) density
  function $f_{\bold Z}(\cdot).$ Suppose the constraints on ${\bf Z}$
  are as follows:
  \begin{enumerate}[(i)]
    \item $\left(v_1(\bold Z),v_2(\bold Z),\ldots,v_{k_1}(\bold
        Z)\right)$ have a joint density function given by $g(\cdot)$.
    \item The moments of $v_{k_1+1}(\bold Z),\ldots,v_{k_2}(\bold Z)$
      are, respectively, $c_1,\ldots,c_{k_2-k_1}.$
    \end{enumerate}
    where $0\leq k_1\leq k_2\leq N$ and
    $v_{1}(\cdot),\,v_{2}(\cdot),\,\ldots,\,v_{k_2}(\cdot)$ are some
    functions on $S$. If the total number of constraints $k_2$ is
    smaller than $N,$ we define $N-k_2$ additional functions
    $v_{k_2+1}(\cdot),\,v_{k_2+2}(\cdot),\,\ldots,\,v_{N}(\cdot)$ such
    that the function $v:S\rightarrow\mathbb{R}^N$ defined by $v(\bold
    z)=\left(v_1(\bold z),v_2(\bold z),\ldots,v_N(\bold z)\right)$ has
    a non-singular Jacobian almost everywhere.  That is,
    \[J(\bold z):=\text{det}\left(\left(\frac{\partial v_i}{\partial
          z_j}\right)_{i,j}\right)\neq 0\,\,\text{for almost all}\,\,
    \bold z\,\,\text{w.r.t.}\,f_{\bold Z}.\] 
    This happens if the function $v$ is locally invertible almost
    everywhere. Now to compute the posterior density that minimizes
    the relative entropy with respect to the prior density $f_{\bf
      Z}(\cdot)$ while satisfying constraints (i) and (ii), we let
    \begin{align*}
      X_i= v_i(\bold Z) \text{ for } i \leq k_1, \bold X = (X_1,
      \ldots, X_{k_1}),\text{ and }\\
      Y_i=v_{k_1+i}(\bold Z) \text{ for } i \leq N-k_1 , \bold Y =
      (Y_1, \ldots, Y_{N-k_1}).
    \end{align*}
    If we use $f(\cdot, \cdot)$ to denote the prior density function
    corresponding to $(\bold X, \bold Y)$ and $w(\cdot)$ to denote the
    local inverse function of $v(\cdot),$ then by the change of
    variables formula for densities,
    \[f(\bold x,\bold y)=f_{\bold Z}\left(w(\bold x, \bold y)\right)
    [J\left(w(\bold x, \bold y)\right)]^{-1}.\]
    Further, the constraints (i) and (ii) translate in terms of
    $(\bold X,\bold Y)$ into:
    \begin{enumerate}[(a)]
    \item $\bold X \,\,\text{have joint density given by}\,\,g(\cdot)$
      and
    \item For $i=1,\ldots,k_2-k_1,$ the expected value of $Y_i$ is
      $c_i.$
    \end{enumerate}
    Setting $k=k_2-k_1$, it follows from Theorem \ref{THM-OPT-SOL}
    that the optimal joint density function of $(\bold X,\bold Y)$ is:
    \[f_{\bmath{\lambda}}(\bold x,\bold y)=
    \frac{e^{\lambda_1y_1+\lambda_2y_2+\cdots+\lambda_ky_k}f(\bold
      x,\bold y)}
    {\int_ye^{\lambda_1y_1+\lambda_2y_2+\cdots+\lambda_ky_k}f(\bold
      x,\bold y)\,d\bold y}\times g(\bold x)\,,\] where $\bold
    \lambda_k$s are chosen such that $\E_{\bmath{\lambda}}[Y_i] = c_i,
    i=1,\ldots,k.$ Again by changing the variables, it follows that
    the optimal density of $\bold Z$ is given by:
    \[\tilde{f}_{\bold Z}(\bold z)=f_{\bmath{\lambda}}(v_1(\bold
    z),v_2(\bold z),\ldots,v_N(\bold z))J(\bold z)\,\,.\,\,\] It can
    be easily seen that the case of Jacobian being identity matrix
    corresponds to no change of variables, and we recover the solution
    to the original optimization problem ${\bf O_1}.$
\end{remark}

\begin{remark}\label{sensitivity}  
  Suppose that a portfolio performance measure of interest is an
  expectation of some random variable $r(\bold X,\bold Y)$ under the
  posterior measure $f_{\bmath{\lambda}}({\bf x,y}).$ Few examples of
  this measure include expected portfolio return, portfolio variance,
  or the probability that the portfolio loss exceeds a threshold.  For
  $\bmath{\gamma} \in \mathbb{R}^k,$ let
  \[ \Pi(\bmath{\gamma}) = \E_{\bmath{\gamma}}[r(\bold X,\bold Y)].\]
  Then computing the sensitivity $\partial
  \Pi(\bmath{\lambda})/\partial c_i $ is of practical interest (recall
  that $c_i$ is specified in Equation \eqref{eqn:constr_4}).  This
  follows through an easy extension of analysis in \cite{avel1}.  Note
  that
  \begin{align*}
    \frac{\partial \Pi(\bmath{\lambda})}{\partial c_i } = \sum_j \left
      ( \frac{\partial \Pi(\bmath{\gamma})}{\partial \gamma_j }\right
    )_{\bmath{\gamma}=\bmath{\lambda}} \frac{\partial
      \lambda_j}{\partial c_i}. 
  \end{align*}
  Further we have $\E_{\bmath{\lambda}}[h_i({\bf X,Y})] = c_i.$
  Following the lines of proof in the Appendix of \cite{buch:kelly},
  it can be verified that
  \begin{align*}
    \frac{\partial \Pi(\bmath{\gamma})}{\partial \gamma_j } &= \E
    \left[\textnormal{Cov}_{\bmath{\gamma}}(r(\bold X,\bold Y),
      h_j(\bold
      X,\bold Y)| \bold X)\right], \text{ and }\\
    \frac{\partial c_i} {\partial \lambda_j} &= \E
    \left[\textnormal{Cov}_{\bmath{\lambda}}(h_i(\bold X,\bold Y),
      h_j(\bold X,\bold Y)| \bold X)\right],
\end{align*}
where
\[\textnormal{Cov}_{\bmath{\gamma}}(r(\bold X,\bold
Y), h_j(\bold X,\bold Y)| \bold X) = \E_{\bmath{\gamma}}[r(\bold
X,\bold Y) h_j(\bold X,\bold Y)| \bold X] -
\E_{\bmath{\gamma}}[r(\bold X,\bold Y)| \bold X)
\E_{\bmath{\gamma}}(h_j(\bold X,\bold Y)| \bold X].\] Let
$V_{ij}=\partial c_i/\partial \lambda_j$ and $V$ be the matrix with
$V_{ij}$ as its entries. Let $U=V^{-1}.$ Then using Implicit Function
Theorem on $\E_{\bmath{\lambda}}[h_i({\bf X,Y})] = c_i,$ we have that
\[\frac{\partial \lambda_j}{\partial c_i}= U_{ij},\]
where $U_{ij}$ is the $(i,j)^{th}$ entry of the matrix $U.$ In
particular,
\[ \frac{\partial \Pi(\bmath{\lambda})}{\partial c_i } = \sum_j \E
\left[\textnormal{Cov}_{\bmath{\lambda}}(r(\bold X,\bold Y), h_j(\bold
  X,\bold Y)| \bold X)\right]U_{ij}.\] Similarly, suppose that the
density function $g(\cdot)$ depends on a parameter $\alpha,$ which we
express as $g_{\alpha}(\cdot),$ then it follows that
\[\frac{\partial \Pi(\bmath{\lambda}) }{\partial \alpha}=
\E_{\bmath{\lambda}}\left[r(\bold X,\bold Y) \frac{ \frac{\partial
      g_{\alpha}(\bold X)}{\partial \alpha}}{g_{\alpha}(\bold
    X)}\right].\] 
\end{remark}

\section{Portfolio Modeling  in Markowitz Framework}
\label{SEC-PORTFOLIO-APP}
In this section we apply the methodology developed in Section~3 to the
Markowitz framework: Namely to the setting where there are $N$ assets
whose returns under the `prior distribution' are multivariate
Gaussian.  Here, we explicitly identify the posterior distribution
that incorporates views/constraints on marginal distribution of some
random variables and moment constraints on other random variables.  As
mentioned in the introduction, an important application of our
approach is that if for a particular portfolio of assets, say an
index, it is established that the return distribution is fat-tailed
(specifically, the pdf is a regularly varying function), say with the
density function $g(\cdot)$, then by using that as a constraint, one
can arrive at an updated posterior distribution for all the underlying
assets.  Furthermore, we show that if an underlying asset has a
non-zero correlation with this portfolio under the prior distribution,
then under the posterior distribution, this asset has a tail
distribution similar to that given by $g(\cdot)$.

Let $(\bold X,\bold Y)=(X_1,X_2,\ldots,X_{N-k},Y_1,Y_2,\ldots,Y_k)$
have a $N$ dimensional multivariate Gaussian distribution with mean
$\bmath{\mu}=(\bmath{\mu}_{\bold x},\bmath{\mu}_{\bold y})$ and the
variance-covariance matrix
\[\bold\Sigma=\left(\begin{array}{cc}\bold\Sigma_{\bold{xx}}&\bold\Sigma_{\bold  
      {xy}}\\ \bold\Sigma_{\bold
      {yx}}&\bold\Sigma_{\bold{yy}}
  \end{array}\right).\]
Let $g(\cdot)$ be a given probability density function on
$\mathbb{R}^{N-k}$ with finite first moments along each component and
$\bold a$ be a given vector in\ $\mathbb{R}^k$. Then we look for a
posterior measure $\tilde{\Pr}(\cdot)$ that satisfies the view that
\[{\bf X} \,\,\text{has probability density function}\,\,
g(\cdot)\,\,\text{and}\,\,\tilde{\E}(\bold Y)=\bold a.\] As discussed
in Remark \ref{change:var} (see also Example \ref{6asset:example} in
Section 5), when the view is on marginal distributions of linear
combinations of underlying assets, and/or on moments of linear
functions of the underlying assets, the problem can be easily
transformed to the above setting by a suitable change of variables.

To find a distribution of $(\bold X,\bold Y)$ which incorporates the
above views, we solve the minimization problem ${\bf O_2}$:
\[\min_{\tilde{f} \in {\cal P}(f)} \int_{(\bold x,\bold
  y)\in\mathbb{R}^{N-k}\times \mathbb{R}^{k}}
\log\left(\frac{\tilde{f}(\bold x,\bold y)}{f(\bold x,\bold y)}\right)
\tilde{f}(\bold x,\bold y)\, d\bold xd\bold y\] subject to the
constraint:
\[\int_{\bold y\in\mathbb{R}^{k}} \tilde{f}(\bold x,\bold y) d\bold
y=g(\bold x) \,\,\,\, \text{ for all } \bold x\] and
\begin{equation}
\int_{\bold x\in\mathbb{R}^{N-k}}\int_{\bold y\in\mathbb{R}^{k}}\bold
y \tilde{f}(\bold x,\bold y) d\bold y d\bold x=\bold a, 
\label{mom:y}
\end{equation}
where\ $f(\bold x,\bold y)$ is the density of $N$-variate normal
distribution denoted by $\mathcal{N}_N(\bmath{\mu},\bold\Sigma)$.

\begin{proposition}
\label{normal:marg}
Under the assumption that $\Sigma_{\bold {xx}}$ is invertible, the
optimal solution to ${\bf O_2}$ is given by
\begin{equation}
  \tilde{f}(\bold x,\bold y)=\tilde{f}(\bold y|\bold x) \times g(\bold
  x) 
  \label{post:marg}
\end{equation}
where $\tilde{f}(\bold y|\bold x)$ is the probability density function
of 
\[\mathcal{N}_{k}\left(\bold a+\bold\Sigma_{\bold
    {yx}}\bold\Sigma_{\bold {xx}}^{-1}(\bold x-\E_g[\bold X])\,,\,
  \bold\Sigma_{\bold{yy}}-\bold\Sigma_{\bold {yx}}\bold\Sigma_{\bold
    {xx}}^{-1}\bold\Sigma_{\bold {xy}}\right)\] where $\E_g[\bold X]$
is the expectation of ${\bf X}$ under the density function $g(\cdot)$.
\end{proposition}

\noindent {\bf Tail behavior of the marginals of the posterior
  distribution:} 
We now specialize to the case where $\bold X$ (also denoted by $X$) is
a real valued random variable so that $N=k+1$, and
Assumption~\ref{ass:ass_u_me} below is satisfied by pdf $g(\cdot).$
Specifically, $(X,\bold Y)$ is distributed as
$\mathcal{N}_{k+1}(\bmath{\mu},\bold\Sigma)$ with
$$\bmath{\mu}^T=(\mu_x,\bmath{\mu}_{\bold y}^T) \,\,\,
\text{and}\,\,\,
\bold\Sigma=\left(\begin{array}{cc}\sigma_{xx}&\bmath{\sigma}_{x\bold
      y}^T\\ \bmath{\sigma}_{x\bold
      y}&\bold\Sigma_{\bold{yy}} \end{array}\right)$$
where\ $\bmath{\sigma}_{x\bold
  y}=(\sigma_{xy_{_1}},\sigma_{xy_{_2}},...,\sigma_{xy_{_k}})^T$\ with
$\sigma_{xy_{_i}}=\textnormal{Cov}(X,Y_i).$ 

\begin{assumption}  
  {\em The pdf $g(\cdot)$ is regularly varying: that is, there exists
    a constant $\alpha>1$ (we require $\alpha >1$ so that $g(\cdot)$
    is integrable) such that \[\lim_{t \rightarrow \infty}\frac{g(\eta
      t)}{g(t)}= \frac{1}{\eta^{\alpha}}\] for all $\eta>0$ (see, for
    instance, \cite{feller:two}). In addition, for any
    $a\in\mathbb{R}$ and $b\in\mathbb{R}^+,$
  \begin{equation}  
    \frac{g(b(t-s-a))}{g(t)} \leq h(s)
    \label{eqn:1101}
  \end{equation}
  for some non-negative function $h(\cdot)$ independent of $t$ (but
  possibly depending on $a$ and $b$) with the property that $\E h(Z) <
  \infty$ whenever $Z$ has a Gaussian distribution.}
\label{ass:ass_u_me}
\end{assumption}

\begin{remark}
  Assumption \ref{ass:ass_u_me} holds, for instance, when $g(\cdot)$
  corresponds to $t$-distribution with $n$ degrees of freedom, that
  is,
  \[g(s)=\frac{\Gamma(\frac{n+1}{2})}{\sqrt{n\pi}\Gamma(\frac{n}{2})}\left(1+\frac{s^2}{n}
  \right)^{-\left(\frac{n+1}{2}\right)}\,,\] Clearly, $g(\cdot)$ is
  regularly varying with $\alpha = n+1$.  To see (\ref{eqn:1101}),
  note that \[\frac{g(b(t-s-a))}{g(t)}=
  \frac{(1+t^2/n)^{(n+1)/2}}{(1+b^2(t-s-a)^2/n)^{(n+1)/2}}\,\,\,.\]
  Putting $t'={bt}/{\sqrt{n}}, s'={b(s+a)}/{\sqrt{n}}$ and $c={1}/{b}$
  we have \[\frac{(1+t^2/n)}{(1+b^2(t-s-a)^2/n)}=\frac{1+c^2t'^2}{1+(t'-s')^2}\,\,\,.\]
  Now (\ref{eqn:1101}) readily follows from the fact that
    \[\frac{1+c^2t'^2}{1+(t'-s')^2}\leq \max\{1,c^2\}+c^2s'^2+c^2|s'|,\] 
    for any two real numbers $s'$ and $t'$.  To verify the last
    inequality, note that if $t'\leq s'$ then
    $\frac{1+c^2t'^2}{1+(t'-s')^2}\leq 1+c^2s'^2$ and if $t'>s'$ then
    \begin{eqnarray*}
      \frac{1+c^2t'^2}{1+(t'-s')^2}=\frac{1+c^2(t'-s'+s')^2}{1+(t'-s')^2}
      &=&
      \frac{1+c^2(t'-s')^2}{1+(t'-s')^2}+c^2s'^2+c^2s'\frac{2(t'-s')}{1+(t'-s')^2}\\
      &\leq&  \max\{1,c^2\}+c^2s'^2+c^2|s'|.
    \end{eqnarray*}
    Note that if $h(x)=x^{m}$ or $h(x)= \exp(\lambda x)$ for any $m$
    or $\lambda$ then the last condition in Assumption
    \ref{ass:ass_u_me} holds.
\end{remark}
From Proposition (\ref{normal:marg}), we note that the posterior
distribution of $(X,\bold Y)$ is $\tilde{f}(x,\bold
y)=g(x)\times\tilde{f}(\bold y|x),$ where $\tilde{f}(\bold y|x)$ is
the probability density function of
\[\mathcal{N}_k\left(\bold
  a+\left(\frac{x-\E_g(X)}{\sigma_{xx}}\right)\bmath{\sigma}_{x\bold
    y},
  \bold\Sigma_{\bold{yy}}-\frac{1}{\sigma_{xx}}\bmath{\sigma}_{x\bold
    y}\bmath{\sigma}_{x\bold y}^t\right),\] where $\E_g(X)$ is the
expectation of $X$ under the density function $g(\cdot).$ Let
$\tilde{f}_{Y_1}(\cdot)$ denote the marginal density of $Y_1$ under
the above posterior distribution.  Theorem~\ref{tail} states a key
result of this section.
\begin{theorem}
\label{tail}
Under Assumption \ref{ass:ass_u_me}, if $\sigma_{xy_{_1}}\neq 0$, then 
\begin{equation}\label{limit:marg}
  \lim_{s\to\infty}\frac{\tilde{f}_{Y_{_1}}(s)}{g(s)}=\left  
    (\frac{\sigma_{xy_1}}{\sigma_{xx}} \right )^{\alpha-1}. 
\end{equation}
\end{theorem}
From (\ref{limit:marg}), we have that
\[\lim_{x \rightarrow
  \infty}\frac{\tilde{P}(Y_1 >x) } {P(X>x)} = \left
  (\frac{\sigma_{xy_{_1}}}{\sigma_{xx}} \right )^{\alpha-1},\] where
$\tilde{P}(\cdot)$ denotes the posterior probability measure
associated with ${\bf Y}$.\\

\subsection{Numerical Experiments}
\label{numeric:exp}
To facilitate visual comparisons, we first consider a small two asset
portfolio model in Markowitz framework in Example \ref{2asset:example}
where we observe how the view that a portfolio has a fat tailed
distribution affects the marginal distribution of the individual
assets.  We then consider a more realistic setting involving a
portfolio of 6 global indices whose VaR (value-at-risk) is
evaluated. The model parameters are estimated from historical data.
We then use the proposed methodology to incorporate a view that return
from one of the index has a $t$-distribution, along with views on the
moments of returns of some linear combinations of the indices.

\begin{example}
\label{2asset:example}
\em{ We consider a small portfolio modeling example involving two
  assets $A_1$ and $A_2$.  We assume that the prior distribution of
  returns $(Z_1, Z_2)$ from assets $(A_1, A_2)$ is bivariate Gaussian.
  Specifically,
\[\left[\begin{array}{c}Z_1\\Z_2\end{array}\right]\sim
\mathcal{N}\left(\left[\begin{array}{c}1\\1\end{array}\right],\left[\begin{array}{cc}9.1&3.0\\  
      3.0&1.1\end{array}\right]\right).\]

Suppose the portfolio management team has the following views on these
securities:
\begin{enumerate}[(i)]
\item A bench mark portfolio consisting of $70\%$ in $A_1$ and $30\%$
  in $A_2$ is expected to generate $1.5\%$ average return, while
  having a much heavier tail compared to a Gaussian distribution.
  This may be modeled as a $t$-distribution with $3$ degrees of
  freedom and mean equaling $1.5\%$.
\item Security $A_2$ will generate $1.5\%$ average return.
\end{enumerate}
Let $X=0.7Z_1+0.3Z_2 \text{ and }Y=Z_2.$ Then the above views
correspond to $X$ having a density function given by
\[g(x)=\frac{2}{2.4120\times\pi\sqrt{3}[1+\frac{1}{3}(\frac{x-1.5}{2.4120})^2]^2},\]
and the expectation of $Y$ being equal to 1.5. Under the prior
distribution we have:
\[\left[\begin{array}{c}X\\Y\end{array}\right]=\left[\begin{array}{cc}0.7&0.3\\ 
    0.0&1.0\end{array}\right] 
\left[\begin{array}{c}Z_1\\Z_2\end{array}\right]\sim
\mathcal{N}\left(\left[\begin{array}{c}1\\1\end{array}\right], 
  \left[\begin{array}{ccc}5.818&2.43\\2.43&1.1\end{array}\right]\right).\]
Therefore we see that ${\sigma_{xx}}=5.818$, $\sigma_{xy}=2.43$ and
$\sigma_{yy}=1.1$, so that
$\sigma_{yy}-\frac{1}{\sigma_{xx}}\sigma_{xy}\sigma_{xy}^t=0.08506.$
In this case $a=1.5$. Therefore
\[a+\left(\frac{x-\E_g(X)}{\sigma_{xx}}\right)\sigma_{xy}=1.5+\frac{(x-1.5)}{5.818}\times
2.43=0.8735+0.4177 x.\] By Proposition \ref{normal:marg}, the
posterior distribution of $(X,Y)$ is given by
\[\frac{3.42876}{\sqrt{2\pi}}\text{exp}\left(-5.8782(y-0.41767
  x-0.8735)^2\right)\times\frac{2}{2.4120\times\pi\sqrt{3}\left(1+\frac{1}{3}
    (\frac{x-1.5}{2.4120})^2\right)^2}.\] Then the posterior
distribution of $(Z_1,Z_2)$ is given by
\[\tilde{f}(z_1,z_2)=\frac{3.42876\times
  0.7}{\sqrt{2\pi}}\text{exp}\left(-5.8782(0.29237z_1-0.8747z_2+0.8735)^2\right)\] 
\[\times\frac{2}{2.4120\times\pi\sqrt{3}\left(1+\frac{1}{3}(\frac{0.7z_1+0.3z_2-1.5}
    {2.4120})^2\right)^2}.\]

\begin{figure}[h]
\begin{center}
  \subfigure[Prior marginal density of $X$ is normal with mean=1 and
  variance=5.818. Posterior density is t with df=3,mean=1.5 and
  scale=$\sqrt{5.818}=2.412$]{
  \includegraphics[width=5cm]{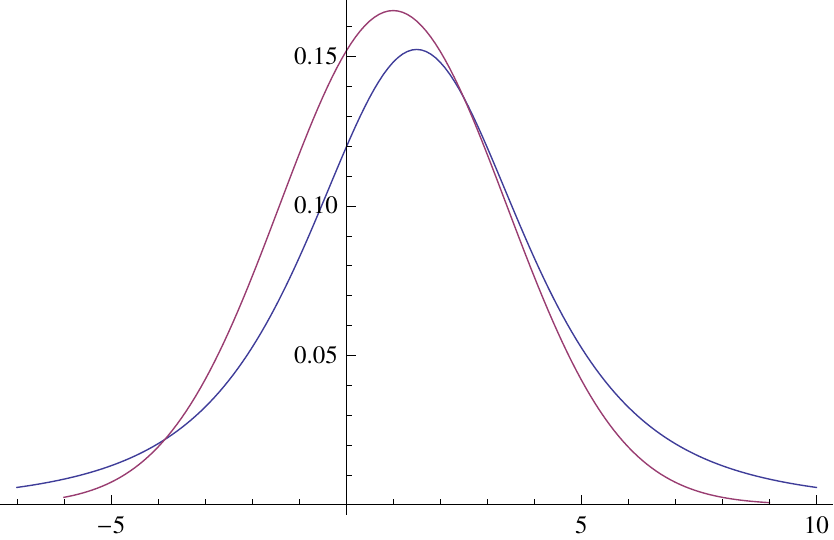}} \hspace{0.3cm}
\subfigure[Prior marginal density of $Z_1$ is normal with mean=1 and
variance=9.1. Posterior density has mean (and mode) 1.5 and heavier
tails .]{
  \includegraphics[width=5cm]{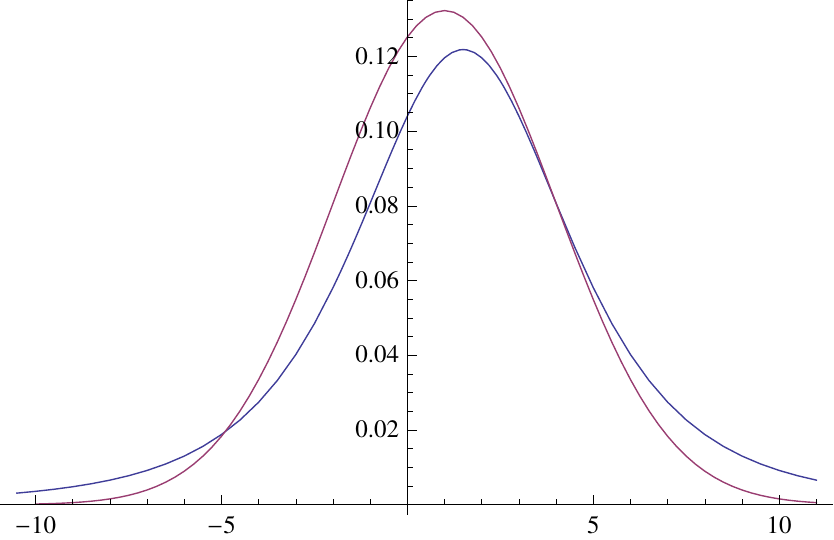}} \hspace{0.3cm}
\subfigure[Prior marginal density of $Z_2$ is normal with mean=1 and
variance=1.1. Posterior density has mean (and mode) 1.5 and heavier
tails.]{
\includegraphics[width=5cm]{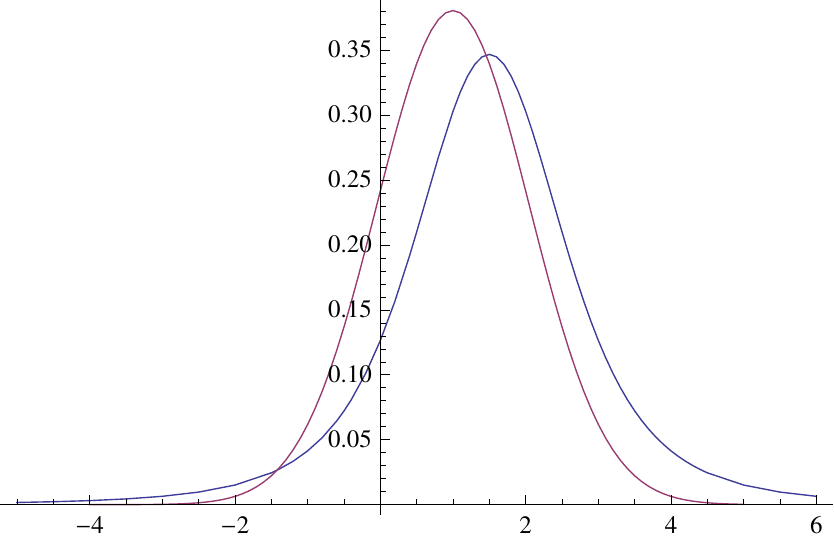}}
\caption{ Prior and posterior marginal densities
  under a constraint on the marginal density of a portfolio.}
\end{center}
\label{fig:marginals}
\end{figure}
In Figure~\ref{fig:marginals} we compare the marginal densities of
$X$, $Z_1$ and $Z_2$ under prior and posterior distributions. We note
that incorporating the constraint that $X$ has a fat-tailed density
renders the asset returns from $A_1$ and $A_2$ to be similarly
fat-tailed.  
}
\end{example}
\begin{example}
\label{6asset:example}
\em{ We consider an equally weighted portfolio in six global indices:
  ASX (the S$\&$P/ASX200, Australian stock index), DAX (the German
  stock index), EEM (the MSCI emerging market index), FTSE (the
  FTSE100, London Stock Exchange), Nikkei (the Nikkei225, Tokyo Stock
  Exchange) and S$\&$P (the Standard and Poor 500).  Let
  $Z_1,Z_2,\ldots,Z_6$ denote the weekly rate of returns from ASX,
  DAX, EEM, FTSE, Nikkei and S$\&$P, respectively.  We take prior
  distribution of $(Z_1,Z_2,\ldots,Z_6)$ to be multivariate Gaussian
  with mean vector 
  \[ [0.062\% \quad 0.28\% \quad 0.045\% \quad 0.13\% \quad 0.24\%
  \quad 0.26\%]\] and variance-covariance matrix
  \[ \left(
    \begin{array}{cccccc} 
      0.4285 & 0.4018 & 0.4394 & 0.3550 & 0.0269 & 0.3194\\
      0.4018 & 0.8139 & 0.6542 & 0.5353 & 0.0558 & 0.5274\\
      0.4394 & 0.6542 & 0.9278 & 0.5248 & 0.0060 & 0.5486\\
      0.3550 & 0.5353 & 0.5248 & 0.4791 & 0.0371 & 0.4220\\
      0.0269 & 0.0558 & 0.0060 & 0.0371 & 0.7606 & 0.0420\\
      0.3194 & 0.5274 & 0.5486 & 0.4220 & 0.0420 & 0.4801
    \end{array}\right) \times 10^{-3}\] 
  estimated from the historical prices of these indices (over the
  period Jan 2010 to Dec 2013). Assuming a notional amount of 1
  million, the historical {value-at-risk} (VaR) and VaR under the
  prior distribution for our portfolio for different confidence levels
  are reported, respectively, in the second and third column of Table
  1.  Next, suppose that we expect the indices ASX, EEM and S$\&$P to
  strengthen and have expected weekly rates of return as $0.1\%, 0.1\%
  \text{ and } 0.35\%$ respectively. Further, consider an independent
  expert view that returns on DAX will exhibit a heavy-tailed
  behaviour. Specifically, let the expert views be:
  \begin{itemize}
  \item[(a)] $\tilde{E}Z_1 = 0.1\%, \tilde{E}Z_3 = 0.1\%, \tilde{E}Z_4
    = 0.13\%, \tilde{E}Z_5 = 0.24\%, \tilde{E}Z_6 = 0.35\%$ and
  \item[(b)] $Z_2$ has a $t$-distribution with 3 degrees of freedom.
  \end{itemize}
  The fourth column in Table 1 reports VaRs at different confidence
  levels under the posterior distribution obtained after incorporating
  views only on expected returns (that xzis, only View (a)).  We see
  that these do not differ much from those under the prior
  distribution. This can be contrasted with the fifth column where we
  have reported the VaRs (computed from 100,000 samples) under the
  posterior distribution obtained after including View (b) that $Z_2$
  has $t$-distribution as well as the View (a) on the expected rates
  of return.
  \begin{table}[h]
    \begin{center}
      \begin{tabular}{|c|p{2cm}|p{2cm}|p{2.5cm}|p{3cm}|p{3cm}|}
        \hline
        VaR at & Historical VaR & Prior distribution & Posterior with
        View (a) & Posterior with Views (a) and (b) & Posterior with Views
        (a) and (c)\\
        \hline
        0.9975 & 67,637 & 56,402 & 56,705 & 79,549 & 67,860\\
        0.9950 & 56,048 & 51,895 & 52,198 & 63,301 & 58,416\\
        0.9925 & 45,524 & 49,099 & 49,402 & 55,853 & 54,067\\
        0.9500 & 29,682 & 33,746 & 34,049 & 29,544 & 33,080\\
        0.7500 & 13,794 & 14,829 & 15,312 & 12,163 & 13,970\\
        0.5000 & \hspace{2pt} 2,836 & \hspace{2pt} 1,680 &
        \hspace{2pt} 2,983 & \hspace{2pt} 1,968 &  \hspace{2pt}
        2,078\\ 
        \hline
      \end{tabular}
      \label{EG-VAR}
      \caption{The second column reports VaR obtained from historical
        returns and the third column reports VaR under the prior
        multivariate normal distribution. While the fourth column is
        for VaR from posterior distribution after including view on
        expected returns (that is, View (a)) the fifth and sixth
        columns corresponds to VaR from posterior distribution after
        incorporating views both on density and expected returns}
    \end{center}
  \end{table}

  Next, suppose that we have the assumption of heavy-tailed density on
  returns of a different asset rather than DAX; for example, consider
  the following view:
  \begin{itemize}
  \item[(c)] A $t$-distribution models returns on S$\&$P
    better. Specifically, let us say that a $t$-distribution with 6
    degrees of freedom is more representative of the tail behaviour of
    observed values of $Z_6.$
  \end{itemize}
  The VaRs corresponding to the posterior distribution obtained after
  incorporating Views (a) and (c) are reported in Column 6 of Table
  1. It can be observed from Columns 5 and 6 of Table 1 that unlike
  posterior distribution which includes views only on expected rates
  of return, marked differences occur from prior VaRs if heavy-tailed
  distribution is assumed for any component asset.

}
\end{example}

\section{Applications to Options Pricing}
\label{SEC-OPT-APP}
In this section, we consider an options pricing scenario where the
implied risk neutral densities of certain highly liquid assets can be
used as views to calibrate models for pricing options written over
assets that may be less actively traded but correlated with the liquid
assets. We show that such views on densities get easily incorporated
through our optimization problem $\bf{O_1}.$

In the option pricing scenario, an effective way to price an option is
through evaluating its expected payoff with respect to the risk
neutral density implied by the option prices observed in the
market. However, as discussed in \cite{breeden1978prices}, estimating
implied risk neutral densities require availability of data over a
large number of strikes. There exists a huge body of literature on
extracting implied risk neutral density from observed option prices of
highly liquid stocks that are traded at many strikes (see
\cite{jackwerth2004option} for a comprehensive review). However, for a
stock which is not actively traded, estimation of implied risk neutral
density is difficult, and in such cases, the implied risk neutral
density available for some heavily traded benchmark asset which is
representative of the market and correlated with the stock of our
interest can be posed as a view/constraint on the marginal
distribution. This view can then be incorporated in the prior
Black-Scholes model to arrive at a posterior model which is more
representative of the observed options prices. To illustrate the
applicability of our framework to the option pricing problem, we
provide an example here.

\begin{example}
{\em Consider the problem of pricing an out of the money
call option on IBM stock trading at USD 82.98 on Jan 5, 2005. This option
with strike at USD 88 is set to expire after 72 days, that is on March 18,
2005. The annual risk free interest rate is $2.69\%.$  We may have
further relevant additional information in the market: For example, if
there is an in the money option at strike price USD 80 on the same IBM
stock for the same maturity which is traded heavily at USD 4.53, this
additional information presents  itself as the following constraint on
risk neutral density: 
\begin{equation} 
  e^{-D} \E [(X-80)^+] = 4.53,
  \label{CONSTRAINT1}
\end{equation}
where $X$ is the value of IBM stock at maturity, $D$ is the
discount factor and $\E[\cdot]$ is the expectation operator with
respect to the risk neutral measure. Further, it is easy to obtain the
daily closing bid and ask prices for some highly liquid instruments
like Standard and Poor's 500 Index options. S\&P500 index is widely
accepted as the proxy for U.S. market portfolio. As mentioned in the
Introduction, its risk neutral density 
implied by the traded options encompasses information that can be
expressed via constraints on expectations (like option prices, mean
rate of return, etc.) and much more (like market sentiments, risk
preferences, sensitivity to new information, etc.); see
\cite{figlewski2009estimating} for a detailed discussion on
this. The S\&P500 option price data and the risk free rate that we
use are from Table 1 in \cite{figlewski2009estimating} (this
facilitates in utilizing the same implied risk neutral density
extracted from the option prices in \cite{figlewski2009estimating}).
Let $Y$ denote the value of S\&P500 index at maturity and
$g_{\textnormal{rnd}}(\cdot)$ denote the risk neutral density of $Y$
implied by the option prices. Then the following constraint gets
imposed on the joint risk neutral density $f_{X,Y}(\cdot,\cdot)$:
\begin{equation}
  \int f_{X,Y}(x,y)\diff x = g_{\textnormal{rnd}}(y), \text{ for all } y.
  \label{CONSTRAINT2}
\end{equation}
For computing the prior joint risk neutral density, we calibrate
multi-asset Black-Scholes model from the historical data of IBM and
S\&P500 prior to Jan 5, 2005 obtained from Yahoo Finance. This results
in a normal prior with covariance matrix
\begin{equation*}
    \begin{bmatrix}
      3.969& -0.4721\\
      -0.4721& 4.489
    \end{bmatrix} \times 10^{-5}
  \label{PRIOR-LN}
\end{equation*}
on log-returns of the IBM stock and S\&P500 index respectively. The
overall problem naturally manifests into finding a posterior density
close to the specified risk neutral lognormal prior $f(\cdot,\cdot)$
while satisfying constraints \eqref{CONSTRAINT1} and
\eqref{CONSTRAINT2}. From Theorem 2, the 
posterior joint risk neutral density
$f_{\textnormal{pos}}(\cdot,\cdot)$ is of the following form:
\[ f_{\textnormal{pos}}(x,y) = e^{\lambda (x - 80)^+}
f(x|y)g_{\textnormal{rnd}}(y),\] where $\lambda$
solving
\[ \int_{x,y} (x-80)^+e^{\lambda (x - 80)^+}
f(x|y) g_{\textnormal{rnd}}(y)\diff x \diff y = 
4.53e^D\] is found to be 0.2479 numerically. This can then be used to
compute the price of the out of the money IBM option with strike at
USD 88 as below:
\[ e^{-D}\int_{x,y} (x-88)^+e^{0.2479 (x - 80)^+}
f(x|y)g_{\textnormal{rnd}}(y)\diff x \diff y = 1.17.\]}
\end{example}
It is easy to incorporate additional option price constraints as in
\eqref{CONSTRAINT1} and find a posterior risk neutral density that is
consistent with all the observed option prices along with the implied
marginal risk neutral density.

\section{Conclusion}
\label{SEC-CONC}
In this article, we built upon the existing methodologies that use
relative entropy based ideas for incorporating mathematically
specified views/constraints to a given financial model to arrive at a
more accurate one. Our key contribution is that we extend the proposed
methodology to allow for constraints on marginal distributions of
functions of underlying variables in addition to moment constraints.
In addition, we specialized our results to the Markowitz portfolio
modeling framework where multivariate Gaussian distribution is used to
model asset returns.  Here, we developed closed-form solutions for the
updated posterior distribution.  In case when there is a constraint
that a marginal of a single portfolio of assets has a fat-tailed
distribution, we showed that under the posterior distribution,
marginal of all assets with non-zero correlation with this portfolio
have similar fat-tailed distribution. This may be a reasonable and a
simple way to incorporate realistic tail behavior in a portfolio of
assets. We also illustrated an application of the proposed framework
in option pricing setting.  Finally, we numerically tested the
proposed methodology on simple examples.

\vspace{0.2in}
{\noindent {\bf Acknowledgement:}
\em The authors would like to thank Paul Glasserman for directional
suggestions  that greatly helped this effort. 
}

\bibliographystyle{abbrv} 
\bibliography{project2011}

\appendix
\section*{Appendix: Proofs}
Here we provide proofs of Theorem \ref{thm:exist_unique}, Proposition
\ref{normal:marg} and Theorem \ref{tail}.\\

\noindent \textbf{Proof of Theorem \ref{thm:exist_unique}:\ \ } Let
$F:\mathbb{R}^k\rightarrow \mathbb{R}$ be a function defined as
\[F(\bmath{\lambda})=\int_{\bold x}\log\left(\int_{\bold y}
  \exp\left(\sum_l\lambda_l h_l(\bold x, \bold y)\right)f(\bold y|
  \bold x)d \bold y\right)g(\bold x)d \bold x-\sum_l\lambda_l c_l.\]
Then,
\begin{eqnarray*}
  \frac{\partial F}{\partial\lambda_i}&=&\int_{\bold
    x}\left(\frac{\int_{\bold y} h_i(\bold x, \bold y)
      \exp\left(\sum_l\lambda_l h_l(\bold x,\bold y)\right)f(\bold
      y|\bold x)d \bold y}{\int_{\bold y}\exp\left(\sum_l\lambda_l
        h_l(\bold x,\bold y)\right)f(\bold y|\bold x)d \bold
      y}\right)g(\bold x)d \bold x - c_i\\  
  &=&\int_{\bold x}\left(\int_{\bold y} h_i(\bold x,\bold y)
    \frac{\exp\left(\sum_l\lambda_l h_l(\bold x,\bold y)\right)f(\bold
      y|\bold x)}{\int_{\bold y}\exp\left(\sum_l\lambda_l h_l(\bold
        x,\bold y)\right)f(\bold y|\bold x)d \bold y}\,d \bold
    y\right)g(\bold x)d \bold x - c_i\\ 
  &=&\int_{\bold x}\left(\int_{\bold y} h_i(\bold x, \bold
    y)f_{\bmath{\lambda}}(\bold y|\bold x)d \bold y\right)g(\bold x)d
  \bold x - c_i\\ 
  &=&\int_{\bold x}\int_{\bold y} h_i(\bold x,\bold
  y)f_{\bmath{\lambda}}(\bold x, \bold y)d \bold xd \bold y - c_i\\   
  &=&\E_{\bmath{\lambda}}[h_i(\bold X,\bold Y)]-c_i .
\end{eqnarray*}
Hence the set of equations given by $\E_{\bmath{\lambda}}[h_i(\bold
X,\bold Y)] = c_i, i=1,\ldots,k$ is equivalent to:
\begin{equation}
\label{eqn:301}
\left(\frac{\partial F}{\partial\lambda_1},\frac{\partial
    F}{\partial\lambda_2},\ldots,\frac{\partial F}{\partial\lambda_k}
\right)=0\,. 
\end{equation}
The solution to this set of equations exist when the prior model is
such that \[\nabla \log \E\left[
  \exp\left(\sum_{l=1}^k\lambda_lh_l(\bf{X,Y})\right)\right] =
(c_1,\ldots,c_k)^T,\] for some $\lambda =
(\lambda_1,\ldots,\lambda_k)$ in $\mathbb{R}^k.$ Since
\[\frac{\partial}{\partial\lambda_j}f_{\bmath{\lambda}}(\bold y|\bold 
x)=h_j(\bold x,\bold y)f_{\bmath{\lambda}}(\bold y|\bold
x)-\left(\int_{\bold y} h_j(\bold x,\bold y)f_{\bmath{\lambda}}(\bold
  y|\bold x)d \bold y\right) \times f_{\bmath{\lambda}}(\bold y|\bold
x),\] we have
\begin{eqnarray*}
  \frac{\partial^2
    F}{\partial\lambda_j\partial\lambda_i}&=&\int_{\bold
    x}\left(\int_{\bold y} h_i(\bold x,\bold
    y)\frac{\partial}{\partial\lambda_j}f_{\bmath{\lambda}}(\bold
    y|\bold x)\,d \bold y\right)g(\bold x)\,d \bold x\\ 
  &=&\int_{\bold x}\left(\int_{\bold y} h_i(\bold x,\bold y)h_j(\bold
    x,\bold y)f_{\bmath{\lambda}}(\bold y|\bold x)d \bold
    y\right)g(\bold x)d \bold x\\ 
  &\,& -\int_{\bold x}\left(\int_{\bold y} h_j(\bold x,\bold
    y)f_{\bmath{\lambda}}(\bold y|\bold x)d\bold
    y\right)\left(\int_{\bold y} h_i(\bold x,\bold
    y)f_{\bmath{\lambda}}(\bold y|\bold x)d \bold y\right)g(\bold x)d
  \bold x\\ 
  &=&\E_{g(\bold x)}\left[\E_{\bmath{\lambda}}[h_i(\bold X,\bold
    Y)h_j(\bold X,\bold Y)\mid \bold X]\right]- 
  \E_{g(\bold x)}\left[\E_{\bmath{\lambda}}[h_j(\bold X,\bold Y)\mid
    \bold X]\times \E_{\bmath{\lambda}}[h_i(\bold X,\bold Y)\mid \bold
    X]\right]\\ 
  &=&\E_{g(\bold x)}\left[\text{Cov}_{\bmath{\lambda}}[h_i(\bold
    X,\bold Y),h_j(\bold X,\bold Y)\mid \bold X]\right],
\end{eqnarray*}
where $E_{g(\bold x)}$ denote expectation with respect to the density
function $g(\bold x)$. By our assumption, it follows that the Hessian
of $F(\cdot)$ is positive definite. Thus, the function $F(\cdot)$ is
strictly convex in $\mathbb{R}^k$.  Therefore if there exists a
solution to (\ref{eqn:301}), then it is unique. Since (\ref{eqn:301})
is equivalent to our constraints that $\E_{\bmath{\lambda}}[h_i(\bold
X,\bold Y)] = c_i,$ the theorem follows.\hfill $\Box$
\vspace{10pt}

\noindent\textbf{Proof of Proposition~\ref{normal:marg}:\ \ } By Theorem
\ref{THM-OPT-SOL}, $\tilde{f}(\bold x,\bold y)=g(\bold
x)\times\tilde{f}(\bold y|\bold x),$ where 
\[\tilde{f}(\bold y|\bold
x)=\frac{e^{\bmath{\lambda}^T\bold y}f(\bold y|\bold x)}{\int
  e^{\bmath{\lambda}^T\bold y}f(\bold y|\bold x)d\bold y}.\] Here the
superscript $T$ corresponds to the transpose.  Now\ $f(\bold y|\bold
x)$ is the $k$-variate normal density with mean vector
$\bmath{\mu}_{\bold y|\bold x}=\bmath{\mu}_{\bold
  y}+\bold\Sigma_{\bold {yx}}\bold\Sigma_{\bold {xx}}^{-1}(\bold
x-\bmath{\mu}_{\bold x})$ and the variance-covariance matrix
$\bold\Sigma_{\bold y|\bold
  x}=\bold\Sigma_{\bold{yy}}-\bold\Sigma_{\bold
  {yx}}\bold\Sigma_{\bold {xx}}^{-1}\bold\Sigma_{\bold {xy}}.$ Hence\
$\tilde{f}(\bold y|\bold x)$\ is the normal density with mean\
$(\bmath{\mu}_{\bold y|x}+\bold\Sigma_{\bold y|\bold
  x}\bmath{\lambda})$\ and variance-covariance matrix\
$\bold\Sigma_{\bold y|\bold x}$. Now the moment constraint equation
(\ref{mom:y}) implies:

\begin{eqnarray*}
  \bold a&=&\int_{\bold x\in\mathbb{R}^{N-k}}\int_{\bold
    y\in\mathbb{R}^{k}}\bold y \tilde{f}(\bold x,\bold y) d\bold y
  d\bold x\\ 
  &=&\int_{\bold x\in\mathbb{R}^{N-k}}g(\bold x)\left(\int_{\bold
      y\in\mathbb{R}^{k}}\bold y \tilde{f}(\bold y | \bold x) d\bold
    y\right) d\bold x\\ 
  &=&\int_{\bold x\in\mathbb{R}^{N-k}}g(\bold
  x)\left(\bmath{\mu}_{\bold y|\bold x}+\bold\Sigma_{\bold y|\bold
      x}\bmath{\lambda}\right)d\bold x\\ 
  &=&\int_{\bold x\in\mathbb{R}^{N-k}}g(\bold
  x)\left(\bmath{\mu}_{\bold y}+\bold\Sigma_{\bold
      {yx}}\bold\Sigma_{\bold {xx}}^{-1}(\bold x-\bmath{\mu}_{\bold
      x}) 
    +\bold\Sigma_{\bold y|\bold x}\bmath{\lambda}\right)d\bold x\\
  &=&\bmath{\mu}_{\bold y}+\bold\Sigma_{\bold {yx}}\bold\Sigma_{\bold
    {xx}}^{-1}(\E_g[\bold X]-\bmath{\mu}_{\bold x}) 
  +\bold\Sigma_{\bold y|\bold x}\bmath{\lambda}.
\end{eqnarray*}
Therefore, to satisfy the moment constraint, we must take
\[\bmath{\lambda}=\bold\Sigma_{\bold y|\bold x}^{-1}\left[\bold
  a-\bmath{\mu}_{\bold y}-\bold\Sigma_{\bold {yx}}\bold\Sigma_{\bold
    {xx}}^{-1}(\E_g[\bold X]-\bmath{\mu}_{\bold x})\right]\,.\]
Putting the above value of\ $\bmath{\lambda}$\ in\
$(\bmath{\mu}_{\bold y|\bold x}+\bold\Sigma_{\bold y|\bold
  x}\bmath{\lambda})$\ we see that $\tilde{f}(\bold y|\bold x)$\ is
the normal density with mean $\bold a+\bold\Sigma_{\bold
  {yx}}\bold\Sigma_{\bold {xx}}^{-1}(\bold x-E_g[\bold X])$ and
variance-covariance matrix\ $\bold\Sigma_{\bold y|\bold
  x}.$\hfill$\Box$
\vspace{10pt}

\noindent \textbf{Proof of Theorem~\ref{tail}:\ \ } We have
\[\tilde{f}(\bold y|x)=D \exp \left(-\frac{1}{2}(\bold
  y-\tilde{\bmath{\mu}}_{\bold y|x})^t\bold\Sigma_{\bold
    y|x}^{-1}(\bold y-\tilde{\bmath{\mu}}_{\bold y|x})\right)\] for an
appropriate constant $D$, where 
\[\tilde{\bmath{\mu}}_{\bold y|x} = \bold
a+\left(\frac{x-E_g(X)}{\sigma_{xx}}\right)\bmath{\sigma}_{x\bold
  y}.\] Suppose that the stated assumptions hold for $i=1$.  Under the
optimal distribution, the marginal density of $Y_1$ is
\[\tilde{f}_{Y_1}(y_1)=\int_{(x,y_2,...,y_k)}D
\exp\left(-\frac{1}{2}(\bold y-\tilde{\bmath{\mu}}_{\bold
    y|x})^t\bold\Sigma_{\bold y|x}^{-1}(\bold
  y-\tilde{\bmath{\mu}}_{\bold y|x})\right)g(x)dxdy_2...dy_k.\] Now
the limit in (\ref{limit:marg}) is equal to:
\[\lim_{y_1\to\infty}\int_{(x,y_2,y_3,...,y_k)}D
\exp\left(-\frac{1}{2}(\bold y-\tilde{\bmath{\mu}}_{\bold
    y|x})^t\bold\Sigma_{\bold y|x}^{-1}(\bold
  y-\tilde{\bmath{\mu}}_{\bold
    y|x})\right)\times\frac{g(x)}{g(y_1)}dxdy_2...dy_k.\]
The term in the exponent is:
\[-\frac{1}{2}\sum_{i=1}^k\left(\bold\Sigma_{\bold
    y|x}^{-1}\right)_{ii}\left((y_i-a'_i)-x\frac{\sigma_{xy_i}}{\sigma_{xx}}\right)
^2 -\frac{1}{2}\sum_{i\neq j}\left(\bold\Sigma_{\bold
    y|x}^{-1}\right)_{ij}\left((y_i-a'_i)-x\frac{\sigma_{xy_{_i}}}{\sigma_{xx}}\right)
\left((y_j-a'_j)-x\frac{\sigma_{xy_{_j}}}{\sigma_{xx}}\right)\] where 
$a'_i=a_i-{E_g(X)}\sigma_{xy_i}/{\sigma_{xx}}.$ Now we make the
following substitutions:
\[(x,y_2,y_3,...,y_k)\longmapsto \bold y'=(y_1',y_2',y_3',...,y_k'),\]
\[y_1'=(y_1-a'_1)-x\frac{\sigma_{xy_{_1}}}{\sigma_{xx}},\text{ and }
y_i'=(y_i-a'_i)-x\frac{\sigma_{xy_{_i}}}{\sigma_{xx}},i=2,3,...,k.\]
Assuming that\ $\sigma_{xy_{_1}}=\textnormal{Cov}(X,Y_1)\neq 0$, the
inverse map $\bold y'=(y_1',y_2',y_3',...,y_k')\longmapsto
(x,y_2,y_3,...,y_k)$ is given by:
\[x=\frac{\sigma_{xx}}{\sigma_{xy_1}}(y_1-y'_1-a'_1), \quad\quad
y_i=y'_i+a'_i+\frac{\sigma_{xy_i}}{\sigma_{xy_1}}(y_1-y'_1-a'_1)
\text{ for }i=2,3,...,k,\]
\[\text{with Jacobian}\ \
\left|\det\left(\frac{\partial(x,y_2,y_3,...,y_k)}{\partial(y_1',y_2',y_3',...,y_k')}
  \right)\right|=\frac{\sigma_{xx}}{\sigma_{xy_{_1}}}.\] The integrand
then becomes
\[D \exp\left(\frac{1}{2}{\bold y'}^T\bold\Sigma_{\bold y|x}^{-1}\bold
  y'\right)\frac{g\left(\frac{\sigma_{xx}}{\sigma_{xy_{_1}}}(y_1-y'_1-a'_1)\right)}{g(y_1)}
\frac{\sigma_{xx}}{\sigma_{xy_{_1}}}.\]
By assumption,
\[\frac{g\left(\frac{\sigma_{xx}}{\sigma_{xy_1}}(y_1-y'_1-a'_1)\right)}{g(y_1)}\leq
h(y_1)\,\,\text{for all}\,y_1,\] for some non-negative function
$h(\cdot)$ such that $Eh(Z) <\infty$ when $Z$ has a Gaussian
distribution. Therefore, by dominated convergence theorem, we have
that
\begin{align*}
  \lim_{y_1\to\infty} &\int {D \exp\left(-\frac{1}{2}{\bold
        y'}^T\bold\Sigma_{\bold y|x}^{-1}\bold
      y'\right)\frac{g\left(\frac{\sigma_{xx}}{\sigma_{xy_{_1}}}(y_1-y'_1-a'_1)\right)}{g(y_1)}
    \frac{\sigma_{xx}}{\sigma_{xy_{_1}}}d\bold y'}\\
  =&\int D \exp\left(-\frac{1}{2}{\bold y'}^T\bold\Sigma_{\bold
      y|x}^{-1}\bold
    y'\right)\lim_{y_1\to\infty}\frac{g\left(\frac{\sigma_{xx}}{\sigma_{xy_{_1}}}
      (y_1-y'_1-a'_1)\right)}{g(y_1)}\frac{\sigma_{xx}}{\sigma_{xy_{_1}}}d\bold 
  y'\\
  =&\int D \exp\left(-\frac{1}{2}{\bold y'}^T\bold\Sigma_{\bold
      y|x}^{-1}\bold
    y'\right)\lim_{y_1\to\infty}\frac{g\left(\frac{\sigma_{xx}}
      {\sigma_{xy_{_1}}}(y_1-y'_1-a'_1)\right)}
  {g(y_1-y'_1-a'_1)}\lim_{y_1\to\infty}\frac{g(y_1-y'_1-a'_1)}
  {g(y_1)}\frac{\sigma_{xx}}{\sigma_{xy_{_1}}} d\bold y',
\end{align*}
which, by our assumption on $g(\cdot)$, in turn equals
\[\int D \exp\left(-\frac{1}{2}{\bold y'}^T\bold\Sigma_{\bold
    y|x}^{-1}\bold
  y'\right)\times\left(\frac{\sigma_{xy_1}}{\sigma_{xx}}\right)^\alpha\times
1\times \frac{\sigma_{xx}}{\sigma_{xy_1}}d\bold
y'=\left(\frac{\sigma_{xy_{_1}}}{\sigma_{xx}}\right)^{\alpha-1}\,\,.\,\,\quad 
\quad\quad\quad\Box\]



\end{document}